\documentclass[10pt,a4paper]{article}
\usepackage{cite}
\usepackage{float}
\usepackage{algorithmic}
\usepackage{graphicx}
\usepackage{textcomp}
\usepackage{xcolor}
\usepackage{mathtools, nccmath,amsmath,amssymb,amsfonts,amstext, dsfont, color}
\usepackage{stmaryrd}
\usepackage{xparse} 
\usepackage{subcaption}
\usepackage{blindtext}
\usepackage[margin=0.95in]{geometry}

\newtheorem{theorem}{Theorem}
\newtheorem{proposition}{Proposition}
\newtheorem{lemma}{Lemma}

\newtheorem{proof}{Proof}

\DeclareMathOperator*{\argmin}{arg\,min}

\newcommand\givenbase[1][]{\:#1\lvert\:}
\newcommand{\Proba}[1]{\mathrm{Pr}\left(#1\right)}

\newcommand{\e}{\mathrm{e}}
\let\given\givenbase

\DeclarePairedDelimiterX\Basics[1](){\let\given\sgiven #1}
\newcommand{\E}[1]{\mathbb{E}\left[#1\right]}

\newcommand{\ItemEstCount}[2]{\hat n_{#1}(#2)} 
\newcommand{\ItemRealCount}[2]{n_{#1}(#2)}

\def\BibTeX{{\rm B\kern-.05em{\sc i\kern-.025em b}\kern-.08em
    T\kern-.1667em\lower.7ex\hbox{E}\kern-.125emX}}

\title{A Formal Analysis of the Count-Min Sketch\\ with Conservative Updates\footnote{This is the author version of Y. Ben Mazziane, S. Alouf, G. Neglia, ``A formal analysis of the count-min
sketch with conservative updates,'' {\it IEEE INFOCOM 2022 - IEEE Conference on Computer Communications Workshops (INFOCOM WKSHPS)}, 2022, pp. 1-6.}}
\author{Younes Ben Mazziane, Sara Alouf, Giovanni Neglia\\ 
        Université C\^{o}te d'Azur, Inria, 
Email: name.surname@inria.fr
        }
\date{}

\begin{document}
\maketitle


\subsection*{Abstract}
Count-Min Sketch with Conservative Updates (CMS-CU) is a popular algorithm to approximately count items' appearances in a data stream. Despite CMS-CU's widespread adoption, the theoretical analysis of its performance is still wanting because of its inherent difficulty. 
In this paper, we propose a novel approach to study CMS-CU and derive new upper bounds on the expected value and the CCDF of the estimation error under an i.i.d. request process. Our formulas can be successfully employed to derive improved estimates for the precision of heavy-hitter detection methods and improved configuration rules for CMS-CU. The bounds are evaluated both on synthetic and real traces.

\section{Introduction}
Counting how many times a given item appears in a data stream is a basic step common to a variety of applications spanning different domains including network management. For example, routers and servers often routinely count the number of packets in each flow for 
troubleshooting, traffic monitoring \cite{basat2017optimal}, detection of denial of service attacks, etc. Similarly, caching policies often rely on content popularity estimates~\cite{einziger2017tinylfu}. Counting is a deceptively simple operation: in many applications the available memory does not permit to instantiate a counter for each possible item, because the  number of items is huge (e.g., catalogs of cacheable objects in content delivery networks) or because counters are updated frequently and then require expensive fast memories (e.g., for high-rate inline packet flow processing). As a consequence, these applications rely on approximate counting techniques
such as sketch-based algorithms \cite{cormode_finding_frequent_items}, among which a popular one is the Count-Min Sketch (CMS) \cite{cormode2005improved}. CMS is also a building brick of more recent sketch algorithms~\cite{yang2018elastic,hsu2019learning}. 


CMS achieves significant memory reduction by mapping different items to the same counters through hash functions. 
As  
different items may increment the same counter, CMS suffers from overestimation errors. 
When counters are only incremented, a slight modification to CMS operation, referred to as Conservative Update~\cite{estan2002new} or Minimal Increment~\cite{cohen2003spectral}, can reduce the estimation error. The Count-Min sketch with Conservative Updates (CMS-CU)  is successfully employed for caching~\cite{einziger2017tinylfu}, heavy flows detection~\cite{wang2021dap}, telemarketing call detection~\cite{bianchi2011demand}, and natural language processing~\cite{goyal2012sketch}.

Although conservative updates are a minor modification to CMS operation, they heavily correlate the growth of the different counters, making CMS-CU much more difficult to study than CMS. As CMS-CU reduces CMS estimation errors, it is still possible to maintain upper bounds originally proposed for CMS~\cite{cormode2005improved,cormode2005summarizing}. This approach has been adopted in some papers, for example to study CMS-CU's trade-off between memory and accuracy~\cite{ventruto2020frequency,wang2021dap}, but  fails to capture the specific advantages offered by CMS-CU.

To the best of our knowledge, only three
papers ventured to study CMS-CU~\cite{bianchi2012modeling,einziger2015formal,chen2021precise}.
The authors of \cite{bianchi2012modeling} relied on a fluid approximation under the assumption that all counters are equally likely to be updated at each step. This assumption may be satisfied only for a large number of counters and a large number of items with similar popularities. 
Reference~\cite{einziger2015formal} modeled CMS-CU as a stack of Bloom filters and derived bounds for the error's Complementary Cumulative Distribution Function (CCDF) when requests follow the Independent Reference Model (IRM)~\cite{irm-fagin-1977}. Unfortunately,  CCDF computation in~\cite{einziger2015formal} is an iterative procedure whose time complexity grows quadratically with the error value.
Moreover, the analysis in both papers hold for families of $k$-wise independent hash functions~\cite{motwani1995randomized}, where $k$ may be arbitrarily large. But such families are incompatible with the memory-constrained applications that need CMS-CU, because  memory requirements and computation time grow with $k$~\cite{siegel2004universal}.
More recently, the authors of \cite{chen2021precise} propose an online algorithm to estimate the error. 

In this paper, we propose a novel analysis of CMS-CU which leads to new upper bounds on the expected value and the CCDF of the estimation error under an IRM request process. 
Our methodology diverges from related work as it quantifies  the error on a per-item basis, which is particularly suited for data streams with heterogeneous items' popularities.
The analysis also overcomes the limitations of the previous studies as 1)~it holds for pairwise independent hash functions, and 2) it provides CCDF expressions with time complexity independent of the error's value. We show that our formulas can be successfully employed to derive improved estimates for the precision of heavy-hitter detection methods and improved configuration rules for CMS-CU.

The rest of the paper is organized as follows. In Sec.~\ref{s:background}, we provide the background and introduce the notation. The theoretical analysis is carried out in Sec.~\ref{s:analysis}. Section~\ref{s:experimental} presents numerical experiments both on synthetic and real world traces.

\section{Background, Notation and Assumptions}
\label{s:background}

\subsection{Data Stream Model} 
A data stream is a sequence $S_t=(Z(s))_{s=1,\ldots,t}$, where $Z(s)$ is an item from a universe $I =\{1, \ldots, N\}$ \cite{book_muthukrishnan2005data}. 
In general, we want to compute a function of the sequence, $\mathcal{F}(S_t)$, for example the number of occurrences of a given item, 
the set of heavy hitters (items whose number of requests exceeds a given threshold), or the top-$k$ most frequent items.
Streaming algorithms aim to compute the function of interest using a few passes through the data stream (only one for the applications we consider)  with an amount of memory which is sublinear in the universe's size $N$ and the data stream size $t$. 
Even for the simple quantities mentioned above, exact computation requires a linear amount of memory and then the streaming algorithms need to settle for approximated results.
In the next section, we present two popular streaming algorithms for approximate counting.

In what follows, we denote the set of integer numbers between $1$ and $d \in \mathbb N$ by $[d]$. 
Moreover, to lighten the notation, we do not append the sketch name to the symbols. We believe there will be no ambiguity as each sketch is presented and analyzed in a separate section.

\subsection{Count-Min Sketch (CMS)}

A Count-Min sketch is a two dimensional array with $d$ rows, each with $w$ counters. An item $i$ is mapped to a counter in each row via $d$ hash functions $\{h_r\}_{r\in [d]}$ chosen uniformly at random from a family of pairwise independent hash functions. 

\begin{equation}
    h_r: \; I    \to   \{1, \ldots , w \} , \;  \forall r\in [d].     
\end{equation}
We note that, once selected, the hash functions do not change during the processing of the stream $S_t$. We model the association between items and counters as a bipartite undirected graph $G=(I,O,E)$, where $O$ is the set of counters and $E\triangleq \left \{(i,h_r(i)):i\in I, r\in [d] \right \}$ is the set of edges. We denote the open neighbourhood of node $i$ in the graph as $N_{G}(i)\triangleq \{c : (i,c) \in E\}$ and the value at time $t$ of the counter in row $r$ corresponding to item $i$ as $c_{i}^{r}(t)$. When item~$i$ is requested at time $t$, the counters $\{h_{r}(i)\}_{r\in [d]}$ are incremented by~$1$. Namely,

\begin{equation}
\label{e:Update_procedure_CMS}
     c_{i}^{r}(t) = c_{i}^{r}(t-1) + 1, \; \forall r\in [d].   
\end{equation}
Let $\ItemRealCount{i}{t}$ denote  item-$i$'s number of occurrences in the stream up to time $t$. Note that $c_i^r(t)$ is  updated not only by new requests for item $i$, but also by requests for all items that are also mapped by $h_r$ to the same counter $h_r(i)$, i.e., by all items in the set $\{j \in I : h_r(j)= h_r(i) \}$. These items are said to \emph{collide} with $i$. It follows that $c_{i}^{r}(t)= \sum_{j:\; h_r(i)=h_r(j)} \ItemRealCount{j}{t}$.
As such, $c_{i}^{r}(t)$ upper bounds~$\ItemRealCount{i}{t}$. We denote the error resulting from using $c_{i}^{r}(t)$ for estimating $\ItemRealCount{i}{t}$ as $e_{i}^{r}(t)$, i.e., $e_{i}^{r}(t)\triangleq c_{i}^{r}(t)-\ItemRealCount{i}{t}$. Since all counters' values $\{c_{i}^{r}(t)\}_{r\in [d]}$ upper bound $\ItemRealCount{i}{t}$, their minimum also upper bounds $\ItemRealCount{i}{t}$. This minimum is the estimate of $\ItemRealCount{i}{t}$ provided by CMS and we denote it as $\ItemEstCount{i}{t}$, 
\begin{equation} 
\label{e:estimation_item}
    \ItemEstCount{i}{t} \triangleq \min_{r\in [d]} c_{i}^{r}(t) ~. 
\end{equation}
The estimation error is then 
\begin{align}
e_{i}(t)&\;\;\triangleq\ItemEstCount{i}{t}-\ItemRealCount{i}{t} = \min_{r\in [d]} \,e_{i}^{r}(t)~.
\label{e:ei-as-min}
\end{align}
We also introduce $\delta_{i,j}^{r}(s)$ to represent the contribution of item $j\neq i$ to  counter $h_r(i)$ at time $s$. 
We have: 
\begin{align}
\label{e:deltaij-CMS}
&\delta_{i,j}^{r}(s) \triangleq \mathds{1}\big(Z(s)=j,\; h_r(i)=h_r(j)\big)~, \\ 
&e_{i}^{r}(t)= \sum_{s\in[t]} \; \sum_{j\in I\setminus\{i\}} \delta_{i,j}^{r}(s)~.
\label{e:eirt-sum-delatij}
\end{align}
All quantities we defined are random variables due to the initial random choice of the hash functions.
From~\eqref{e:eirt-sum-delatij} and the definition of pairwise independence~\cite{motwani1995randomized}, one can immediately conclude that  $\E{e_{i}^{r}(t)}= \frac{\sum_{j\neq i} \ItemRealCount{j}{t} }{w} \leq \frac{t}{w}$. Applying~\eqref{e:ei-as-min}, we obtain the following upper bound for the expected estimation error:
\begin{equation}
\label{e:ub-expected-error-CMS}
    \E{e_{i}(t)} \leq \frac{t}{w}~.
\end{equation}

Moreover, the random variables $\{e_{i}^{r}(t)\}_{r\in [d]}$ are i.i.d., and an application of the Markov inequality leads to the following
upper bound on the CCDF 
of $e_{i}(t)$:
    \begin{equation}
    \label{e:ub_CCDF_error_CMS_agnostic}
        \Proba{\frac{e_{i}(t)}{t} \geq x} \leq \left(\frac{1}{wx}\right)^{d} ~.
    \end{equation}
Cormode and Muthukrishnan proved this result in {\cite[Theorem 1]{cormode2005improved}} 
for the particular value $x=\frac{\e}{w}$. 


\subsection{Count-Min Sketch with Conservative Updates (CMS-CU)}
\label{s:backgroun-cms-cu}

The conservative update \cite{estan2002new} or minimal increment \cite{cohen2003spectral} is an optimization of CMS that consists in incrementing only the counters that attain the minimum value. The update procedure 
when item $i$ is requested at time $t$ becomes 

\begin{equation}
\label{e:Update_procedure_CMS_CU}
 c_{i}^{r}(t) = \max \left( c_{i}^{r}(t-1) , \min_{f \in [d]}  c_{i}^{f}(t-1) + 1    \right), \; \forall r\in [d]~.\!\!
\end{equation}
The error $e_{i}^{r}(t)$ in each row $r$, the estimation count $\ItemEstCount{i}{t}$, and the estimation error $e_{i}(t)$, all depend on $c_{i}^{r}(t)$ in the same way as in CMS. Equations \eqref{e:estimation_item} and \eqref{e:ei-as-min} hold with CMS-CU. The quantities $\{\delta_{i,j}^{r}(s)\}_{s\in[t], j\in I\setminus \{i\}}$ are now defined as 
    \begin{equation}
    \label{e:deltaij-CMS_cu}
        \delta_{i,j}^{r}(s) \triangleq \mathds{1}\big(Z(s)=j, h_r(i)=h_r(j), \ItemEstCount{j}{s-1}=c_{i}^{r}(s-1)\big). 
    \end{equation}
Equation~\eqref{e:eirt-sum-delatij} holds for CMS-CU. With respect to \eqref{e:deltaij-CMS}, \eqref{e:deltaij-CMS_cu} captures the additional condition that  counter $h_r(i)$ is updated by a request for $j$ at time $s$ only if its current value $c_i^r(s-1)$ coincides with the current estimate $\hat n_j(s-1)$.
Because of this additional condition, CMS-CU enjoys always a smaller error than CMS. Therefore, 
 CMS upper bounds on the  expectation~\eqref{e:ub-expected-error-CMS} and on the CCDF~\eqref{e:ub_CCDF_error_CMS_agnostic} also hold for CMS-CU.


\subsection{Our Assumptions} 
\label{s:assump}
We will assume in our analysis that the request process follows the Independent Reference Model (IRM) \cite{irm-fagin-1977}, in other words,   $\{Z(s)\}_{s\in[t]}$ are i.i.d.~categorical random variables with $\Proba{Z(s)=i}=p_i$, for $i \in I$, and $\sum_{i\in I} p_i=1$. We refer to~$p_i$ as the \emph{popularity} of item $i$. Without loss of generality, we number items in $I$ according to their popularity rank, hence $p_i \geq p_{i+1}$, for $i\in [N-1]$. Note that there are two sources of randomness in our setting: the hash functions' selection and the request process $S_t$.
From now on, the expectation $\E{.}$ and the probability $\Proba{.}$ take both kinds of randomness into account.


\section{Theoretical Analysis of CMS-CU}
\label{s:analysis}

Under the IRM model, we first prove a tighter upper bound on the CCDF of $e_{i}(t)$ for CMS, then we upper bound the expectation and CCDF of $e_{i}(t)$ for CMS-CU.

\subsection{CMS: CCDF of the Estimation Error}
\label{ss:CMS-CCDF}

In this section we will derive a tighter bound for CMS error under the IRM assumption. As discussed in Sec.~\ref{s:backgroun-cms-cu}, this new bound also applies to the CMS-CU error.
We first observe from~\eqref{e:deltaij-CMS} that $\E{\delta_{i,j}^{r}(s)}\leq p_j/w$, since item $j$ is requested with probability $p_j$ and the hash collision probability between $i$ and $j$ equals  $1/w$ because of pairwise independence. It readily follows from \eqref{e:eirt-sum-delatij} that $\E{e_{i}^{r}(t)}\leq (1-p_i)t/w$. Therefore, $(1-p_i)t/w$ is an upper bound on the expected error $\E{e_{i}(t)}$.

\begin{proposition}[Upper bound on the CCDF of $e_{i}(t)/t$]
\label{prop:CCDF-CMS}
The CCDF of the estimation error $\e_{i}(t)$, when using CMS, verifies 
\begin{equation}
\label{e:U.B c.c.d.f error CMS final}
    \Proba{e_{i}(t)/t\geq x} \leq  \mathcal{A}(x)^{d}~,
\end{equation}
\begin{align}
\label{e:Ax_expression}
\hspace{-1.1em}\mbox{where}\qquad\quad    &\mathcal{A}(x) \triangleq \min_{k=0,\ldots,w-1  } \mathcal{A}_k(x)~, \quad \mathcal{A}(0)=1~, \\
\hspace{-1.1em}\mbox{and}   \qquad\qquad  &\mathcal{A}_k(x)  \triangleq \frac{\sum_{j>k}  p_j}{(w-k)x} + \frac{k}{w}~.
\end{align}

\end{proposition}

\begin{proof}
From~\eqref{e:ei-as-min} and the fact that the random variables $\{e_{i}^{r}(t)\}_{r\in [d]}$ are i.i.d.~when using CMS, we have $\Proba{e_{i}(t)/t\geq x} = \left(\Proba{e_{i}^{1}(t)/t \geq x }\right)^{d}$. To prove \eqref{e:U.B c.c.d.f error CMS final} it is then sufficient to show that $\Proba{e_{i}^{1}(t)/t\geq x}\leq \mathcal{A}_k(x)$ for $k=0,\ldots,w-1$. For a given $k\neq0$ we consider the event, called $E_{i,k}^r$, of no hash collision in row $r$ between item $i$ and any of the $k$ most popular items (other than $i$, if $i\leq k$). By first writing the law of total probabilities with respect to the partition $\{E_{i,k}^1, \overline{E_{i,k}^1}\}$, and then using the union bound to write $\Proba{\overline{E_{i,k}^1}}\leq k/w$ and the Markov inequality to upper bound $\Proba{e_{i}^{1}(t)/t\geq x\given E_{i,k}^1}$, we obtain
\begin{align}
\Proba{e_{i}^{1}(t)/t\geq x} &\;\leq \E{e_{i}^{1}(t)\given E_{i,k}^1} / (xt) + k/w 
\label{e:bound-firststep}\\
&\hspace{-4em}\leq  \sum_{j>k} p_j \;\Proba{h_1(i)=h_1(j)\given E_{i,k}^1} / x + k/w
\label{e:bound-Ee1}\\
&\hspace{-4em}\leq  \sum_{j>k} p_j  /((w-k)x) + k/w
\label{e:bound-compEk}
\end{align}
where \eqref{e:bound-Ee1} follows from \eqref{e:deltaij-CMS}-\eqref{e:eirt-sum-delatij} and \eqref{e:bound-compEk} uses $\Proba{\overline{E_{i,k}^1}}\leq k/w$. 
By observing that \eqref{e:bound-compEk} holds also for $k=0$, we have completed the proof.
\end{proof}

Proposition \ref{prop:CCDF-CMS} extends known results in the literature. In particular, upper bounding the right-hand side of \eqref{e:U.B c.c.d.f error CMS final} by $(\mathcal{A}_{0}(x))^d$ yields \eqref{e:ub_CCDF_error_CMS_agnostic}, and then replacing  $x=e/w$, we obtain
{\cite[Theorem 1]{cormode2005improved}}.

In order to highlight the importance of this proposition, we present an example where the improvement of \eqref{e:U.B c.c.d.f error CMS final} over \eqref{e:ub_CCDF_error_CMS_agnostic} is evident. Consider a distribution where the most popular item is requested with probability $1-\epsilon$, and other items have uniform popularity, i.e.,  $p_1= 1 - \epsilon$, $p_i= \epsilon/(N-1), \; \forall i>1$. We compare the two bounds for $x=1/w$. Equation~\eqref{e:ub_CCDF_error_CMS_agnostic} provides the trivial bound $\Proba{e_{i}(t)/t\geq 1/w}\leq 1$. On the other hand, bounding the right-hand side of \eqref{e:U.B c.c.d.f error CMS final} by $(\mathcal{A}_1(x))^{d}$ yields the upper bound 
$\left( \epsilon \frac{w}{w-1} + \frac{1}{w} \right)^{d}$. 
For large $w$ and $\epsilon = o\left(\frac{1}{w}\right)$, we conclude that $\Proba{e_{i}(t)/t\geq 1/w} \lesssim ~\frac{1}{w^{d}} \approx 0$, in sharp contrast with the bound in \eqref{e:ub_CCDF_error_CMS_agnostic}.

\subsection{CMS-CU: Expected Estimation Error}
We consider now CMS-CU and 
derive an upper bound on the expectation of the estimation error. Because of \eqref{e:ei-as-min}, an upper bound on $e_i^r(t)$ suffices. Thus, we turn our attention to the random variable $\delta_{i,j}^r(s)$.
As for CMS, it is easy to prove for CMS-CU that $\E{\delta_{i,j}^{r}(s)} \leq p_j/w$. In the next lemma, we derive a tighter bound, in particular for $j>i$.

\begin{lemma}[Upper bound on $\E{\delta_{i,j}^{r}(s)}$]
\label{lem:ub_expectation_deltaij}
The expected contribution of item $j$ to item $i$'s count at row $r$ at time $s$ satisfies \eqref{e:ub_expectation_deltaij}.  
\begin{align} 
\label{e:ub_expectation_deltaij}
\exists &\alpha_{i,j}>0,\; \beta_{i,j}\geq 0: \nonumber\\
&\quad \E{\delta_{i,j}^{r}(s)}  \leq \frac{p_j}{w} \big(\gamma_{i,j} +  \beta_{i,j} \exp\big(-\alpha_{i,j} (s-1)\big) \big),
\end{align}
with\vspace{-1.5em}
\begin{align}
\label{e:gamma-ij}
\gamma_{i,j}\triangleq 
\begin{cases}
1,                                                & \forall j\leq i \\
\min\left(  \mathcal{A}(p_i-p_j)^{d-1},1 \right), & \forall j>i 
\end{cases}
\end{align}
and $\mathcal{A}(x)$ given in \eqref{e:Ax_expression}.
\end{lemma}
\begin{proof}[Proof of Lemma \ref{lem:ub_expectation_deltaij}]
We will make use of two quantities to prove Lemma \ref{lem:ub_expectation_deltaij}.
\begin{align}
 &   l_j^r \triangleq \sum_{e\in N_{G}(h_{r}(j))} p_e, \qquad   g_j \triangleq \min_{r\in [d]} l_{j}^{r}~.
\label{e:gj}
\end{align}
For a given realization of $G$, $l_j^r$ is an upper bound on the growth rate of counter $c_j^r(t)$ and $g_j$ is an upper bound on the growth rate of $\ItemEstCount{j}{t}$. 
To ease the writing, we use $A$, $B$, $C$, and $D_j^r$ as shorthand for events ``$h_r(i)=h_r(j)$'', ``$\ItemEstCount{j}{s-1}=c_i^r(s-1)$'', ``$g_j\geq p_i$'', and ``$l_j^r\geq p_i$'', respectively. Starting from \eqref{e:deltaij-CMS_cu} we write
\begin{align*}
&\E{\delta_{i,j}^r}=p_j \Proba{A\cap B}\nonumber\\
&\quad = p_j\Big[\Proba{A\cap B\cap C} 
           + \Proba{A\cap B\cap \overline{C}} \Big]
            \nonumber\\
&\quad \leq p_j\Big[\Proba{A\cap C} 
            + \Proba{A\cap \overline{C}}  \Proba{B\given A, \overline{C}}  \Big]
           \nonumber\\
&\quad \leq p_j\Big[\Proba{A\cap \left(\cap_{e\in[d],e\neq r}D_j^e\right)}  
            + \Proba{A}  \Proba{B\given A, \overline{C}}  \Big]
          \nonumber\\
&\quad \leq p_j\Proba{A} \Big[ \Proba{D_j^1}^{d-1} 
                    + \Proba{B\given A, \overline{C}}  \Big]
                      \nonumber\\
&\quad \leq \frac{p_j}{w}\Big[\gamma_{i,j} + \beta_{i,j}\exp(-\alpha_{i,j}(s-1))\Big]
    \nonumber
\end{align*}
where $\gamma_{i,j}$ is given in \eqref{e:gamma-ij}. We obtained the last step by using (for $j>i$)
   $\Proba{D_j^1}= \Proba{l_j^1-p_j\geq p_i-p_j} \leq {\cal A}(p_i-p_j)$,
that can be derived following the steps in \eqref{e:bound-firststep}-\eqref{e:bound-compEk}. The inequality 
\begin{align}
    \Proba{B\given A, \overline{C}} \leq \beta_{i,j}\exp(-\alpha_{i,j}(s-1))
    \label{e:using-chernoff}
\end{align}
also used in the last step, requires more explanations. Due to space constraints we only sketch its derivation. For $z=s-1$, we define the random variable $y_{j}(z)$ as, 

\begin{equation}\label{def:yjt}
     y_{j}(z) \triangleq \sum_{e: h_{r_0}(j) =h_{r_0}(e)} \ItemRealCount{e}{z}: \; r_0 = \argmin_{r\in [d]} l_{j}^{r}~.
\end{equation}
It follows that $\ItemEstCount{j}{z}\leq y_{j}(z)$. Furthermore, since $c_{i}^{r}(z)\geq \ItemRealCount{i}{z}$, we get $\Proba{B\given A, \overline{C}}\leq 1- \Proba{F\given A,\overline{C}}$, where $F=\{\ItemRealCount{i}{z} > y_{j}(z) \}$. Since $n_{i}(z)$ and $y_{j}(z)$ are negatively associated \cite{joag1983negative}, we have:
\begin{flalign} \label{e:neg-binom}
    \Proba{B\given A,\overline{C}} \leq 1 - \Proba{L} \Proba{J \given A,\overline{C}}  
\end{flalign}
Where $L=\{ \ItemRealCount{i}{z} > m(z) \} $, $J= \{ y_j(z) < m(z) \}$ and $m(z)= (p_i+g_j)z/2$. Following \eqref{def:yjt}, for every fixed graph realization of $G$, $y_{j}(z)$ is the sum of negatively associated random variables \cite{joag1983negative} and has an expected value of $g_jz < m(z)$ under the conditioning $g_j<p_i$, thus using Chernoff bounds on events $L$ and $J$ we get \eqref{e:using-chernoff}.

\end{proof}



\paragraph*{Comments on Lemma \ref{lem:ub_expectation_deltaij}}
We recall that using CMS, this expectation is upper bounded by $p_j/w$, thus the term $\gamma_{i,j}+~\beta_{i,j}\exp\left(-\alpha_{i,j} (s-1)\right)$ is an attenuation term taking into account the conservative update. As $s\to\infty$, this attenuation term converges to $\gamma_{i,j}$. The larger the difference between $i$ and $j$ probabilities, the smaller is $\gamma_{i,j}$. This is expected, as the larger the difference in popularity between two items $i$, $j$,  the  more likely that $c_{i}^{r}(t)> \ItemEstCount{j}{t}$.

We now state the main result of this section.

\begin{theorem}[Upper bound on $\E{e_{i}(t)}$]
\label{th:main}
The error experienced by item $i$ is upper bounded as follows
\begin{flalign}
 \exists B_{i}\in & \mathds{R}^{+}:  \; \E{\frac{e_{i}(t)}{t}}  \leq  \frac{1}{w} \sum_{j\in I \setminus \{i\}} p_j \gamma_{i,j} + \frac{B_{i}}{t}~,
\label{e:main}
\end{flalign}
where $\gamma_{i,j}$ is defined in \eqref{e:gamma-ij}. 
\end{theorem}

\begin{proof}
We prove the upper bound on $\E{e_i^r(t)}$ and because of \eqref{e:ei-as-min} this upper bound holds for $\E{e_{i}(t)}$ too. An upper bound on $\E{e_i^r(t)}$ is readily found by linearity of the expectation, using \eqref{e:eirt-sum-delatij} and Lemma \ref{lem:ub_expectation_deltaij}. We find
\begin{align} 
\label{e:ub-expected-error-CMS-cu}
     \E{e_{i}^{r}(t)}
&\leq  \sum_{j\in I\setminus \{i\}} \frac{p_j \gamma_{i,j}}{w} 
t + \frac{\beta_i }{w(1-\exp\left(-\alpha_i\right) )},
\end{align}
where $\alpha_i = \min_{j\in I} \alpha_{i,j}$ and $\beta_i = \max_{j\in I}\beta_{i,j}$ and $\alpha_{i,j}$ and $\beta_{i,j}$ are the constants in Lemma \ref{lem:ub_expectation_deltaij}. 
\end{proof}
As $t\to\infty$, the term $B_i$ can be ignored. Note that the 
bound depends on the item's rank~$i$.
As discussed before, $\gamma_{i,j}$ is a decreasing function of $p_i-p_j$, thus, of $j$. A necessary and sufficient condition to improve over the bound in \eqref{e:ub-expected-error-CMS} for a given item $i$ is then to have $\gamma_{i,N}<1$. At the same time, $\gamma_{i,j}$ is an increasing function of $i$. Therefore, the more popular the item, the smaller the bound \eqref{e:main}, which is always smaller than or equal to the bound \eqref{e:ub-expected-error-CMS} when neglecting $B_i$. While previous studies \cite{einziger2015formal, bianchi2012modeling} bounded the error uniformly across items, our analysis provides error bounds depending on item's popularity. In particular, our work is the first to support analytically the experimental evidence that the most popular items barely experience any error \cite{bianchi2012modeling}. 

To highlight the improvement of our bound over the CMS bound $\E{e_{i}(t)}\leq (1-p_i)/w$ (see the beginning of Sec.~\ref{ss:CMS-CCDF}), we consider the same example as in Sec.~ \ref{ss:CMS-CCDF}, i.e., $p_1=1-\epsilon$, $p_i = \epsilon/(N-1)$ for $i>1$, and focus on the most popular item. 
According to CMS analysis, $\E{e_{1}(t)/t}\leq \epsilon/w$, whereas \eqref{e:ub-expected-error-CMS-cu} yields a bound for the most popular item that is $\epsilon \mathcal{A}\left(1-\epsilon \frac{N}{N-1}\right)^{d-1}/w \approx \epsilon \mathcal{A}(1-\epsilon)^{d-1}/w$. This bound is smaller than $\epsilon \mathcal{A}_{0}(1-\epsilon)^{d-1}/w=\epsilon \left(w(1-\epsilon)\right)^{1-d}/w$. By choosing $\epsilon=\frac{1}{2}$, we get an improvement by a factor of $(w/2)^{d-1}$. 


Having analysed the expectation of the estimation error, we turn our attention to its CCDF, which is studied in the next section.


\subsection{CMS-CU: CCDF of the Estimation Error}

\begin{proposition}[Upper \\bound on the CCDF of $e_{i}(t)/t$]
\label{prop:CCDF-CMS-cu}
The CCDF of the estimation error $e_{i}(t)$, when using CMS-CU, is upper bounded as follows:

\begin{equation} \label{prop:e:main}
  \Proba{\frac{e_{i}(t)}{t}\geq x} \leq \min \left(\mathcal{A}(x)^{d} , \mathcal{B}(x,i,t)\right)~,
\end{equation}
where 
\begin{align}
\label{e:bfunction} \nonumber
    &\mathcal{B}(x,i,t) \triangleq\\
    &\min_{k=0,\ldots,w-1} \frac{1}{x}\left(\frac{1}{w-k} \sum_{j>k} p_j \gamma_{i,j} + \frac{B_{i}(k)}{t} \right) 
     + \frac{k}{w},
\end{align}
and $B_{i}(k)$ are constants that depend only on item $i$ and on $k$.
\end{proposition}

\begin{proof}[Sketch of the proof]
The bounds that are valid with CMS are also valid with CMS-CU, thus by Proposition \ref{prop:CCDF-CMS}, the CCDF with CMS-CU is less than $\mathcal{A}(x)^{d}$. To prove the other part, we rely on the same arguments used in the proofs of Proposition \ref{prop:CCDF-CMS} and Lemma \ref{lem:ub_expectation_deltaij}. We first write \eqref{e:bound-firststep} for the random variable $e_{i}^{r}(t)$. To bound $\E{e_{i}^{r}(t)\given E_{i,k}^r}$, we repeat the derivations in the proof of Lemma \ref{lem:ub_expectation_deltaij} to bound the conditional expectation of $\delta_{i,j}^r(s)$, and then using \eqref{e:eirt-sum-delatij} we obtain
$\E{e_{i}^{r}(t)\given E_{i,k}^r} \leq \sum_{j>k} p_j \gamma_{i,j} t/(w-k) +B_i(k) $.
As this bound is valid for $k=0,\ldots,w-1$, we find \eqref{prop:e:main} which concludes the proof.
\end{proof}

In practical situations, $t$ is large enough such that we can ignore the constants $B_{i}(k)$ in \eqref{e:bfunction} and the bound \eqref{prop:e:main} depends solely on $x$ and $i$.

We will illustrate the utility of Proposition \ref{prop:CCDF-CMS-cu} in the next section where we estimate a metric of interest in the heavy-hitters detection problem.

\subsection{Heavy-Hitters Use Case: Lower Bound on the Precision}
\label{s:precision}

Detecting heavy-hitters in a stream can be done using a sketch. A heavy-hitter is an item that has request rate higher than a threshold $\phi$. However, when using a sketch (for instance CMS or CMS-CU), an item with a rate smaller than~$\phi$ can erroneously appear as a heavy-hitter because of the overestimation error; we call such an item a ``false positive.'' Let $H$ be the set of heavy-hitters, $H= \{ i: \; \ItemRealCount{i}{t} \geq \phi \cdot t \}$, and $\hat H$ be the set of items classified as heavy-hitters by the sketch, $\hat H = \{ i: \; \ItemEstCount{i}{t} \geq \phi \cdot t \}$. The ``precision'' is one metric used for assessing the performance of the sketch, and is defined as follows: $P = |H|/ |\hat H|$. For sake of simplicity, we assume that $\ItemRealCount{i}{t} \approx p_i t$, 
for $t$ large enough, this is reasonable because of the law of large numbers. Under this approximation $|H|$ is constant and we can write the expected value of the precision as: 
\begin{align} 
\label{e:ub_precision}
    \E{P}
& \approx \frac{|H|}{|H|+\sum_{i > |H|} \Proba{ \frac{e_i(t)}{t} \geq \phi -p_i }} ~. 
\end{align}
Combining \eqref{e:ub_precision} with Proposition \ref{prop:CCDF-CMS-cu} we obtain a lower bound on the expected precision when CMS-CU is used. This lower bound will be illustrated in Section \ref{s:exp-precision} and compared to experimental values.

\section{Experimental Evaluation and Numerical Analysis} 
\label{s:experimental}

\subsection{Experimental Setting}

To support our analysis, we have undertaken a series of experiments in which we simulated requests for items over time and used CMS-CU to count the requests for each item. We considered two settings in our experiments. In the first setting, we generated $10^{3}$ synthetic streams from two different Zipf distributions, with shape parameter $\alpha =0.8$ and $\alpha=1.0$. Each stream contains 1 million requests for items in the set $I=\{1, \ldots, 10^{6}\}$ ($N=10^{6}$). We selected different hash functions for each stream by choosing uniformly at random $d$ different seeds in $[10^{6}]$. 
The experimental values reported for this setting are averaged over the $10^{3}$ streams. We also computed the $95\%$ confidence intervals but do not report them as they are very narrow and would hardly be visible in the figures.


In the second setting, we used a trace of accesses to Wikipedia pages in all languages during September 2007~\cite{wiki_trace}. The trace contains 10,628,125 requests. The number of distinct Wikipedia pages requested in this trace is 1,712,459. We extracted 10 non-overlapping stream from this trace, each containing  $10^{6}$ requests, and discarded the rest.

For each of the settings, we report the results obtained for two metrics: $(i)$ the expected estimation error of the sketch for each item, and $(ii)$ the precision in the heavy-hitters detection problem.

\subsection{The Expected Estimation Error}
\label{s:exp_estimation_error}
In the synthetic setting, 
 we computed for each item the average estimation error over the $10^{3}$ streams,
our upper bound 
as in Theorem \ref{th:main} (but neglecting the constant $B_i$), and 
the state-of-the-art bound~\eqref{e:ub-expected-error-CMS} originally proposed in~\cite{cormode2005improved}. Results are shown in Fig.~\ref{plot:error_items}.

\begin{figure}[H]
  \centering
\begin{subfigure}{0.35\linewidth}
  \centering
  \includegraphics[width=0.99\linewidth]{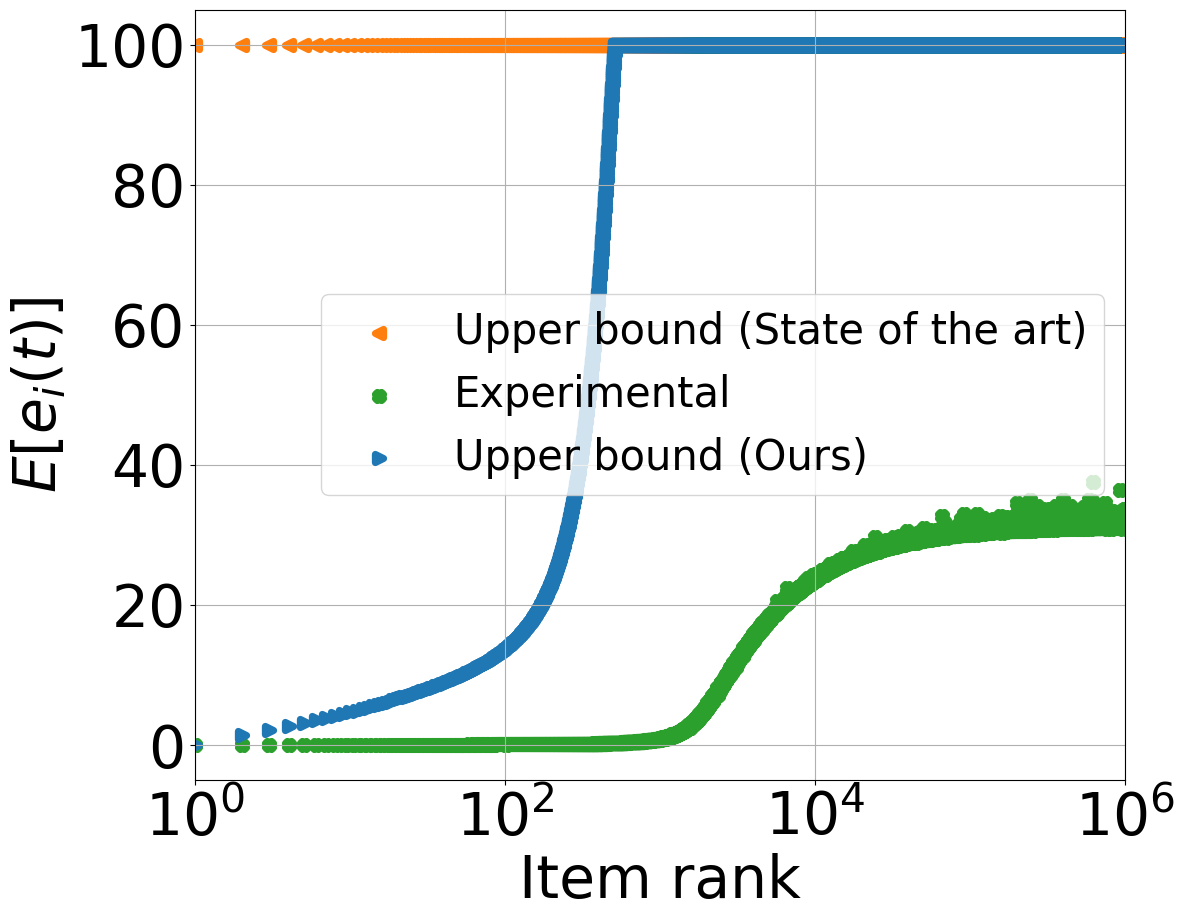}  
  \caption{$\alpha=0.8$, $w=10,000$}
  \label{plot:error_items_a0.8}
\end{subfigure}
\hspace{3em}
\begin{subfigure}{0.35\linewidth}
  \centering
  \includegraphics[width=0.99\linewidth]{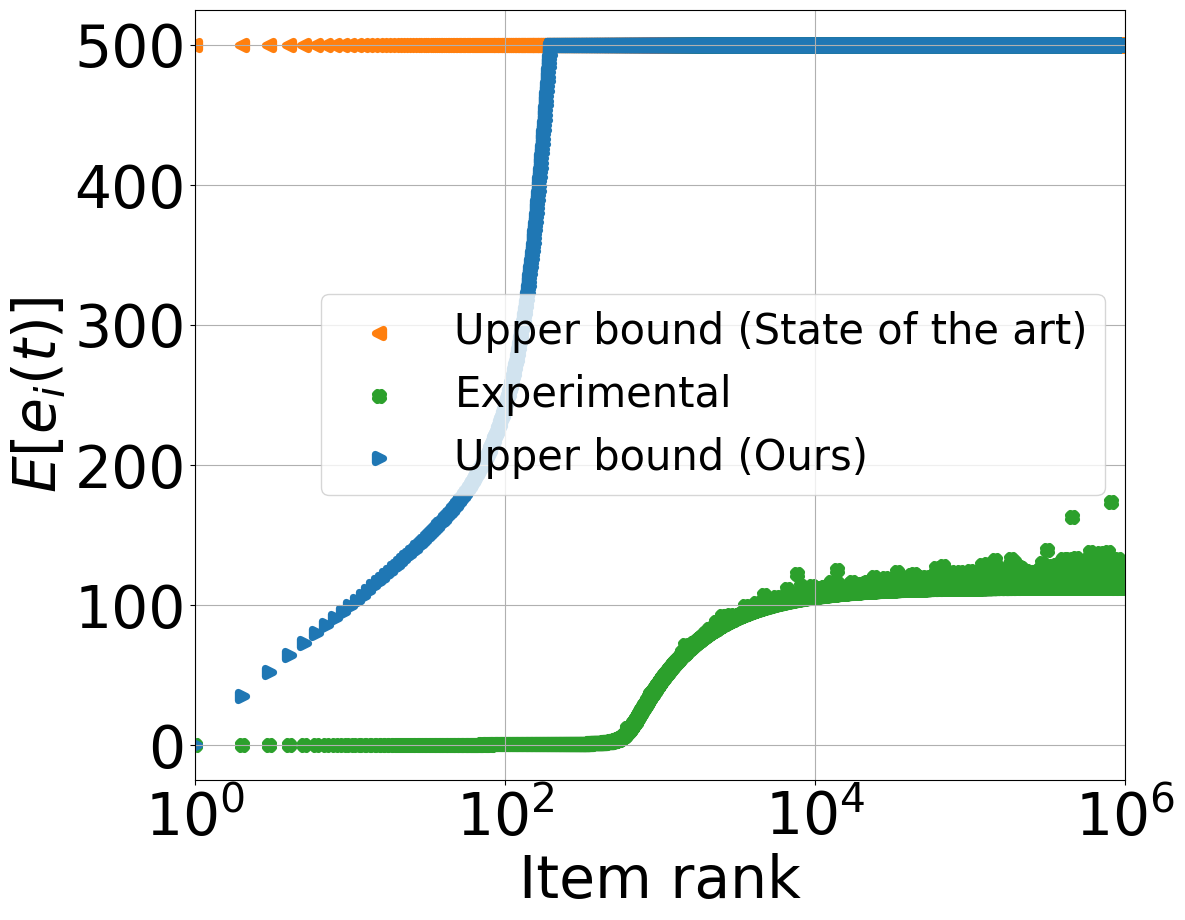}  
  \caption{$\alpha=1.0$, $w=2000$}
  \label{plot:error_items_a1}
\end{subfigure}%
\caption{Synthetic traces: estimation error for each item,\\$N=10^{6}$, $d=5$, $t=10^{6}$.}
\label{plot:error_items}
\end{figure}

As we observed above, our analysis correctly predicts that different items experience a different error and improve the current bound on the expected estimation error for the $500$ and $190$ most popular items for $\alpha=0.8$ and $\alpha=1$, respectively.


For Wikipedia access trace, we used the first stream to estimate items' popularity distribution. Results in Fig.~\ref{plot:error_items_trace} are then averages over the remaining 9 streams computed as done for Fig.~\ref{plot:error_items}.
The curves confirm qualitatively the observations on the synthetic traces: our approach leads to a smaller bound for the most popular items.


\subsection{Precision in Detecting $\phi-$Heavy-Hitters}
\label{s:exp-precision}
As we discussed in Sec.~\ref{s:precision}, the CCDF upper bound in Proposition~\ref{prop:CCDF-CMS-cu} allows us also to derive a lower bound on CMS-CU's expected precision using  approximation~\eqref{e:ub_precision}. 
If the same approximation is combined with the CCDF upper bound~\eqref{e:ub_CCDF_error_CMS_agnostic} from~\cite{cormode2005improved}, one can obtain an analogous lower bound on CMS-CU's expected precision. This lower bound is labelled ``lower bound (State of the art)'' in Figs. \ref{plot:precision_width_d5} and \ref{plot:precision_trace}.

In the experiments with the synthetic streams, we averaged the precision values obtained over the $10^{3}$ runs. We have repeated these experiments for multiple width values. The experimental values are depicted in Fig.~\ref{plot:precision_width_d5} together with our lower bound and the state of the art bound.
Consistently with what observed for the CCDF, our approach improves also precision estimation. The bound becomes tighter for larger values of the width $w$.

\begin{figure}[H]
  \centering
\begin{subfigure}{0.35\linewidth}
  \centering
  \includegraphics[width=0.99\linewidth]{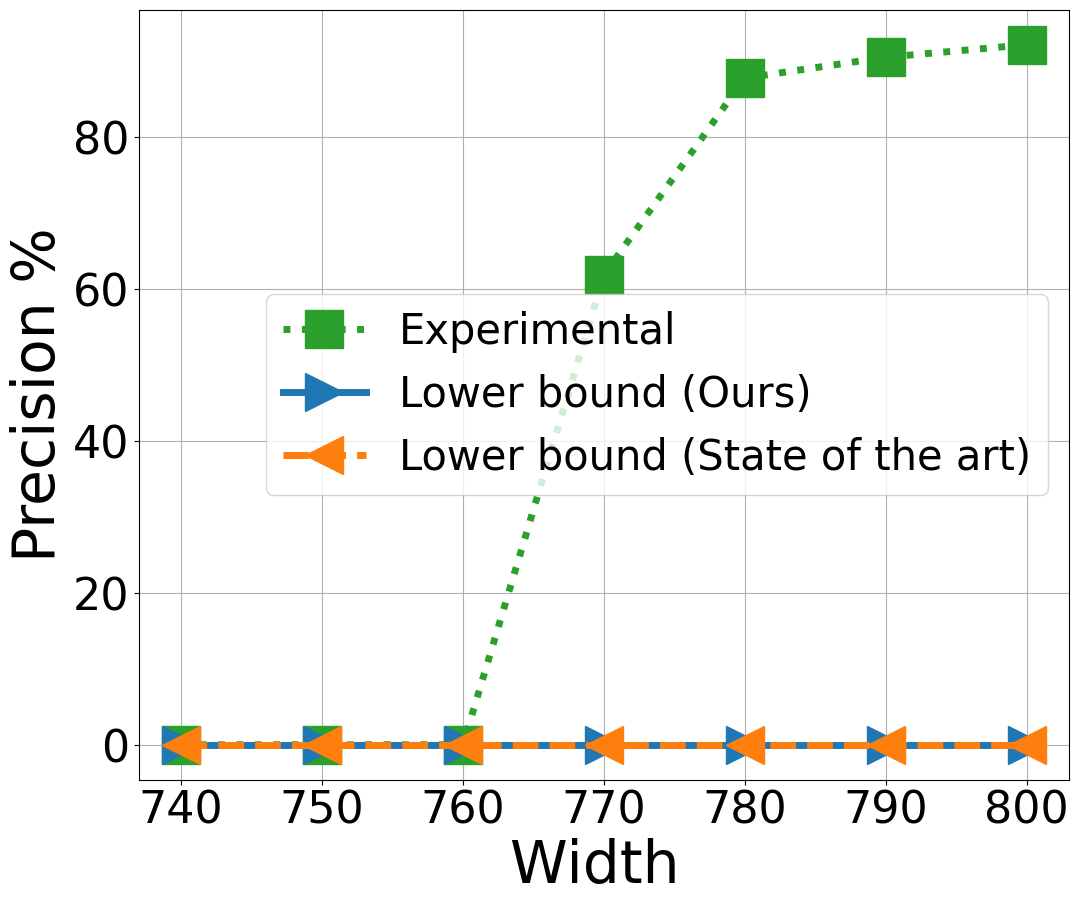}  
  \caption{Small values of the width}
  \label{plot:precision_small_width}
\end{subfigure}%
\hspace{3em}
\begin{subfigure}{0.35\linewidth}
  \centering
  \includegraphics[width=0.99\linewidth]{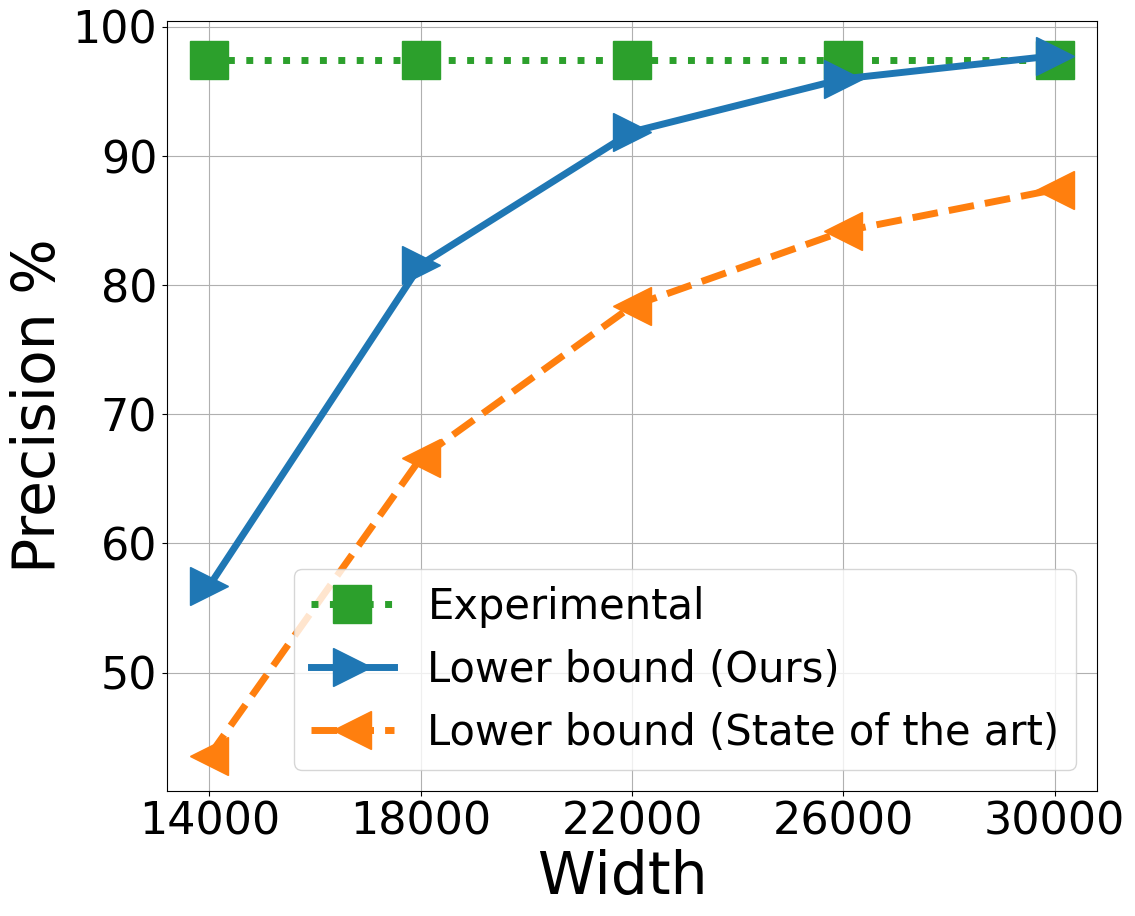}  
  \caption{Large values of the width}
  \label{plot:precision_larger_width}
\end{subfigure}%
\caption{Synthetic trace: Precision as a function of the width $w$, $N=10^{6}$, $\alpha =0.8$, $d=5$, $t=10^{6}$, $\phi = 5\times 10^{-4}$.}
\label{plot:precision_width_d5}
\end{figure}

In the experiments with Wikipedia trace, we slightly changed the popularity estimation procedure in comparison to Sec.~~\ref{s:exp_estimation_error}. The empirical distribution over a stream was used as input to the analytical formulas to predict the precision in the following stream.
Fig.~\ref{plot:precision_trace} shows the corresponding results. The advantage of our approach is even more evident over this real trace. Note that the state-of-the-art CCDF bound in~\eqref{e:ub_CCDF_error_CMS_agnostic} depends on the count sketch parameters $w$ and $d$ but not on the popularity distribution. Nevertheless, the corresponding precision in Fig.~\ref{plot:precision_trace} changes across streams: the approximated formula for the precision \eqref{e:ub_precision} depends on the specific stream because popularity distribution (and then also the number of heavy hitters $|H|$) change from one stream to the other.


\begin{figure}[H]
  \centering
\begin{subfigure}{0.35\linewidth}
  \centering
  \includegraphics[width=0.99\linewidth]{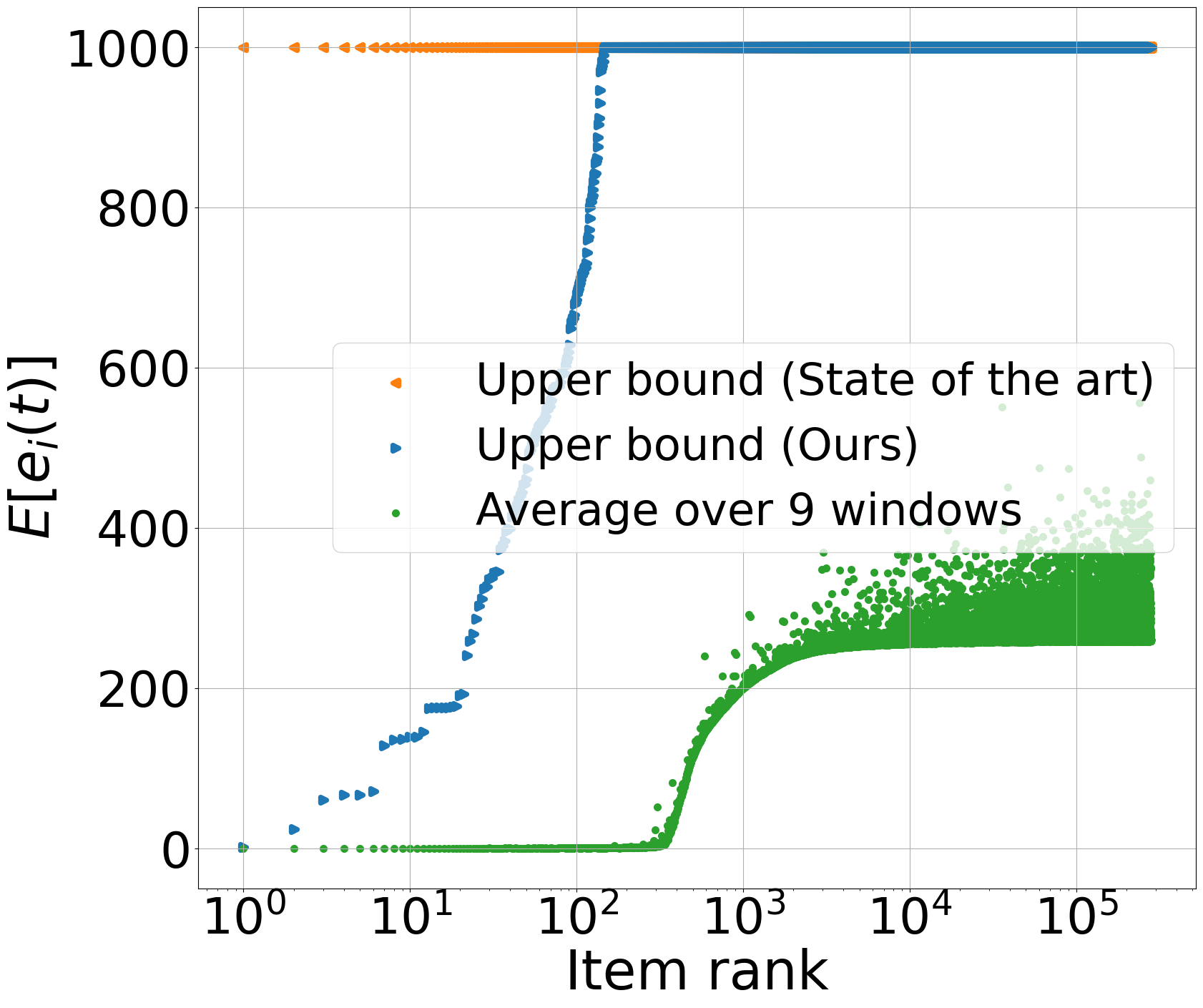}
  \caption{$w=1000$, $d=4$.
  \label{plot:error_items_trace}}
\end{subfigure}
\hspace{3em}
\begin{subfigure}{0.35\linewidth}
  \centering
  \includegraphics[width=0.99\linewidth]{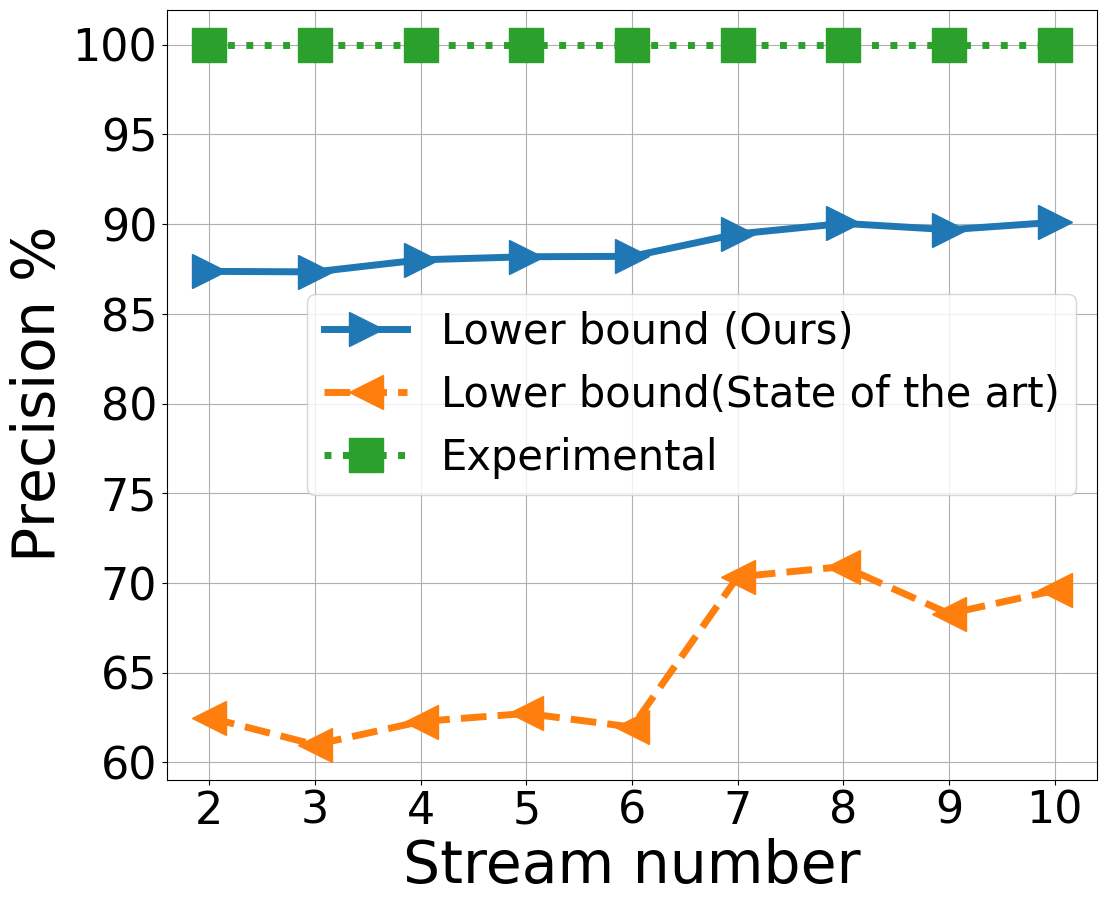}  
  \caption{$w=2000$, $d=6$, $\phi=3\times 10^{-3}$.
  \label{plot:precision_trace}}
\end{subfigure}
\caption{Real trace: estimation error for each item (left) and precision (right), $t=10^{6}.$}
\label{plot:trace}
\end{figure}

\subsection{Configuring CMS-CU with QoS Guarantees} 
The bounds we derived can also be used to configure the width $w$ and the depth $d$ of  CMS-CU in order to achieve the desired precision with the minimum amount of memory.
If each counter uses $4$ bytes, the memory cost of a CMS-CU is $M=4wd$ bytes.
We compared numerically the memory requirements determined by our approach and by the state-of-the-art one.
In particular, for target precision values in the range  0.8--0.975, we performed a search for memory values between $20$ Kbytes and $3.2$ Mbytes (with a step of $20$~Kbytes) and depth values between $2$ and $8$ (the width is then determined as $w=M/4d$) to find the smallest memory which guarantees the target precision. 
Figure~\ref{plot:precision_memory} shows the corresponding curves obtained using our approach and the state-of-the-art-one for the synthetic and Wikipedia trace. 

Our approach leads to configuring CMS-CU using a reduced amount of memory, e.g., for $97.5\%$ target precision, the improvement factors are $6.63$ and $7.12$ for the synthetic and the Wikipedia trace, respectively.


\begin{figure}[H]
  \centering
\begin{subfigure}{0.35\linewidth}
  \centering
  \includegraphics[width=0.99\linewidth]{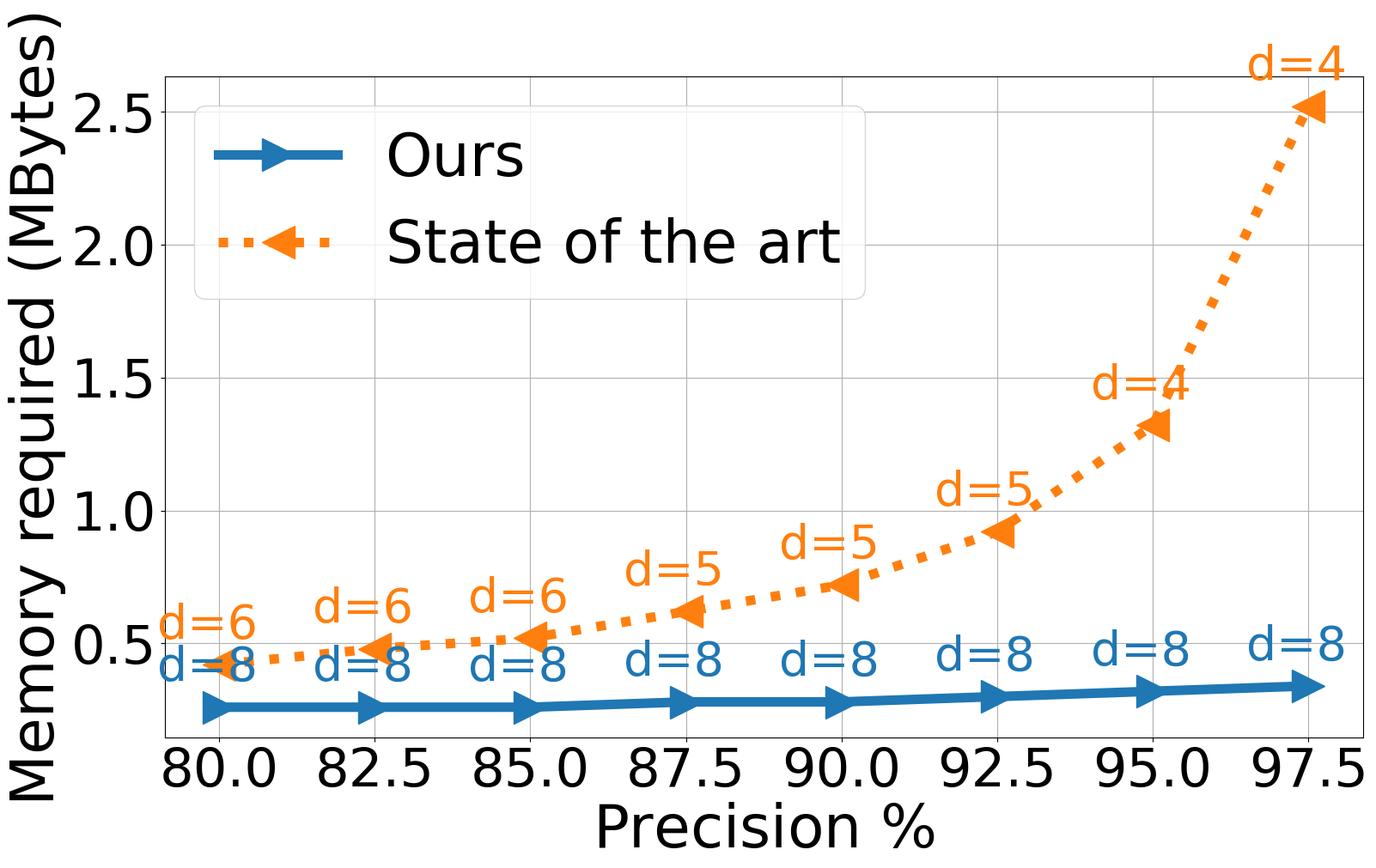}
  \caption{$N=10^{6}$, $\alpha=0.8$
  \label{plot:precision_memory_0.8}}
\end{subfigure}
\hspace{3em}
\begin{subfigure}{0.35\linewidth}
  \centering
  \includegraphics[width=0.99\linewidth]{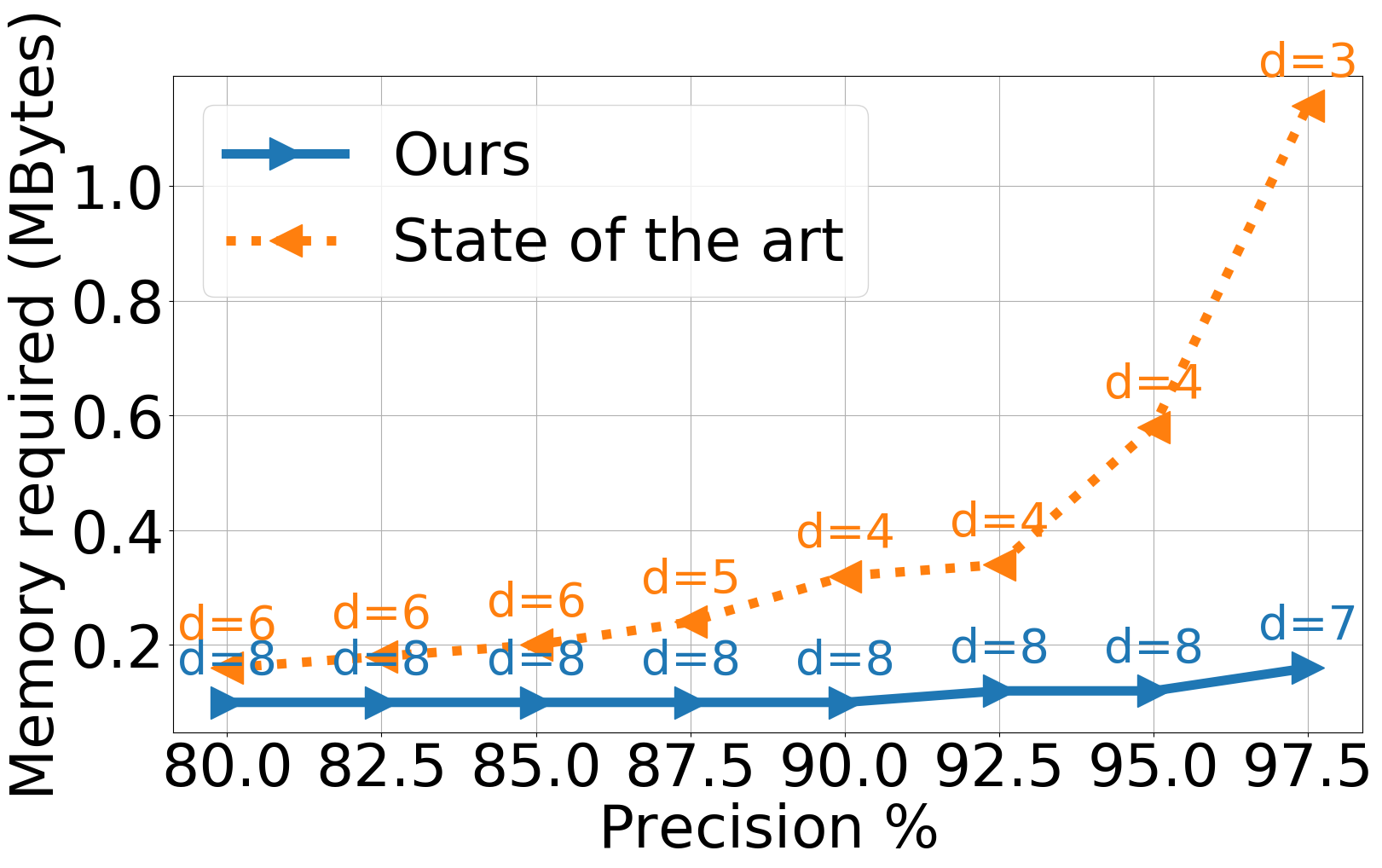}  
  \caption{Wikipedia
  \label{plot:precision_memory_tace}}
\end{subfigure}
\caption{Memory required for a given precision: (a) $\phi=5 \cdot 10^{-4}$, (b) $\phi = 10^{-3}$. }
\label{plot:precision_memory}
\end{figure}

\section{Conclusion and Perspectives}
\label{s:conclusion}
While it is a common belief that CMS-CU leads to smaller estimation errors for the most popular items~\cite{bianchi2012modeling}, our paper is the first to provide quantitative support for such property, thanks to a per-item study of the estimation error. We showed that our analysis significantly improves existing bounds for the most popular items and leads, in comparison to the state of the art,  to more accurate estimations for the precision in heavy-hitter detection problems as well as to improved configuration rules, which avoid to oversize the counting data structure.

For less popular items, our bounds are not tighter than existing ones. In the future, we want then to focus on improving the bounds for the tail of the popularity distribution. A possible approach is to combine our analysis with existing methods to estimate the CMS-CU error floor when items have similar popularities like those in \cite{bianchi2012modeling}.


\bibliography{citing.bib}
\bibliographystyle{plain}

\end{document}